\documentclass[10pt]{amsart}
\usepackage{amssymb}
\usepackage{amsmath}
\usepackage{mathrsfs}
\usepackage{amsthm}
\usepackage{graphicx}
\usepackage{multirow}
\usepackage{url}
\usepackage{subfigure,verbatim}
\usepackage{hyperref}
\usepackage{cite}
\usepackage{bbm}
\usepackage[usenames,dvipsnames]{color}
\usepackage{pgfplots}
\usepgfplotslibrary{groupplots}

\newcommand{\abs}[1]{\left\vert#1\right\vert}

\newcommand{\paren}[1]{\left(#1\right)}
\newcommand{\bracket}[1]{\left[#1\right]}

\newcommand{\R}{\mathbb{R}}

\numberwithin{equation}{section}

\DeclareMathOperator{\MSE}{MSE}
\DeclareMathOperator{\Bias}{Bias}

\DeclareMathOperator{\hi}{hi}
\DeclareMathOperator{\lo}{lo}
\DeclareMathOperator{\Pois}{Pois}
\DeclareMathOperator{\Var}{Var}
\DeclareMathOperator{\Cov}{Cov}
\newcommand{\tsim}{t_{\rm sim}}
\DeclareMathOperator{\argmin}{argmin}

\newcommand{\E}{\mathbb{E}}

\newtheorem{thm}{Theorem}[section]
\newtheorem{theorem}[thm]{Theorem}
\newtheorem{algorithm}[thm]{Algorithm}
\newtheorem{definition}[thm]{Definition}

\title{Analysis of Estimators for Adaptive Kinetic Monte Carlo}
\date{\today}
\author{D. Aristoff}
\author{S. Chill}
\author{G. Simpson}
\thanks{GS was supported by US Department of Energy Award
  DE-SC0012733.  GS also thanks P. Hitczenko for helpful discussions.}

\begin{document}

\maketitle
\begin{abstract}
  Adaptive Kinetic Monte Carlo combines the simplicity of Kinetic
  Monte Carlo (KMC) with a Molecular Dynamics (MD) based saddle point
  search algorithm in order to simulate metastable systems.  Key to
  making Adaptive KMC effective is a stopping criterion for the saddle point search. In this work, we examine a criterion, recently
  appearing in \cite{Chill:2014es}, that is based on the fraction of total 
  reaction rate found instead of the fraction of observed saddles. The criterion uses the Eyring-Kramers law to estimate the 
  reaction rate at the MD search 
  temperature. We also consider a
  related criterion that remains valid when the Eyring-Kramers law is not. We examine the
  mathematical properties of both estimators and prove their mean
  square errors are well behaved, vanishing as the simulation
  continues to run.
\end{abstract}

\section{Introduction}

An outstanding problem in theoretical materials science and chemistry
is how to reach laboratory time scales of microseconds ($10^{-6}$ s)
and longer 
using Molecular Dynamics (MD) based models which resolve the atomistic
time scale of femtoseconds ($10^{-15}$ s).  Much of this scale
separation is due to the presence of {\it metastable} regions in the
configuration space of the system.  In such regions, often defined by
local minima of an energy landscape, the system stays close to a
particular configuration, such as a local minima, before crossing into
some other metastable region associated with a different
configuration.  Consequently, during much of a direct MD simulation,
the system is close to one metastable region or another. It exhibits dynamics akin to a
continuous time random walk on the set of metastable states, with
comparatively long waiting times.

Since much of the physical significance of these systems is
characterized by the sequence of visited metastable states and the
time spent in each, there have been a variety of efforts to
systematically coarse grain the MD trajectory into a more
computationally efficient continuous time random walk.  A.F. Voter has
proposed three methods, Parallel Replica Dynamics, Hyperdynamics, and
Temperature Accelerated Dynamics, which can overcome metastability
through intelligent usage of the primitive Langevin
dynamics,~\cite{perez2009accelerated,Voter:2002p12678}.  In recent
years, significant effort has been made to understand and quantify the
approximations in these methods and extend their
applicability,~\cite{Aristoff:2014ch,Aristoff:2014tm,Aristoff:2014pr,Binder:2015gu,LeBris:2012et,Lelievre:2013ud,Simpson:2013cs}.

Another approach to the problem is Kinetic Monte Carlo (KMC), and
this will be the focus of this work.  Let us assume our system is
governed by a potential energy $V(x)$, $x\in \R^{d}$ at inverse
temperature $\beta$.  Furthermore, we assume that we have partitioned
configuration space into an at most countable set of metastable
states, $\Omega_i$, associated with local minima $m_i$ of $V$.  The
system can go from metastable state $i$ to metastable state $j$ if
there is a saddle point, $s_{ij}$, of $V(x)$ joining $\Omega_i$ and
$\Omega_j$.  {For conciseness, we will assume there is a single saddle
  point joining two given adjacent metastable states, though, in
  general, there may be multiple pathways.}

In traditional KMC, before a simulation is run, one must identify the
metastable states, their connectivity ({\it i.e.}, which ones are
joined by saddle points), and the reaction rates of each such
connection.  Given all of this information, KMC is very cheap to
simulate.  A single random number is generated and used to select one
of the possible reactions, the system migrates into the new metastable
region, and the algorithm repeats.

Unfortunately, such complete details of the metastable states and
their connectivity are, {\it a priori}, unavailable in all but the
simplest low dimensional systems.  This has motivated the development
of Adaptive Kinetic Monte Carlo (AKMC),
\cite{Chill:2014es,Xu:2008jk,Xu:2009fy}.  In AKMC, the system starts
in some metastable region $\Omega_i$.  Saddle points associated with
$\Omega_i$ are then sought via a {\it saddle point search algorithm}
that successively finds $s_{ij}$.  Reaction rates for each such saddle
can be estimated by the Eyring-Kramers law \cite{Hanggi:1990en}:
\begin{equation}
  \label{e:kramers}
  k_{ij} = g_{ij} \exp\left[-\beta (V(s_{ij}) - V(m_i)\right],
\end{equation}
where, writing $\lambda_1$ for the 
sole negative eigenvalue of $\nabla^2 V(s_{ij})$, 
\begin{equation*}
g_{ij} = \frac{|\lambda_1|}{\pi}\sqrt{\left|\frac{\text{det}\nabla^2 V(m_i)}{\text{det}\nabla^2 V(s_{ij})}\right|}.
\end{equation*}
Once a sufficient number of saddles
associated with $\Omega_i$ have been identified, the problem is
treated by using traditional KMC with the thus far identified
reactions and their rates; this process then repeats in the next
metastable region. Two things are needed to proceed with AKMC:
\begin{enumerate}
\item A saddle point search algorithm;
\item A stopping criterion.
\end{enumerate}
In this work, we will consider the question of the stopping criterion,
provided our saddle point search algorithm satisfies certain
assumptions. Our analysis will focus on estimators similar to the
one introduced by Chill \& Henkelman in~\cite{Chill:2014es}.  We call these {\em
  Chill type estimators}.

In \cite{Chill:2014es}, the authors searched for saddle points out of each
metastable state using high temperature MD.  For concreteness, consider
the Brownian dynamics in ${\mathbb R}^d$:
\begin{equation}
  \label{e:od}
  dX_t = -\nabla V(X_t) dt + \sqrt{2\beta^{-1}}dW_t.
\end{equation} 
The aim is to model the dynamics 
at low temperature $\beta = \beta^{\lo}$. Starting at $X_0 \in \Omega_i$, integrate \eqref{e:od} at a 
higher temperature 
$\beta = \beta^{\hi}$ ({\em i.e.}, 
$\beta^{\lo} > \beta^{\hi}$) until the trajectory leaves $\Omega_i$.  Using
the higher temperature $\beta^{\hi}$ allows an escape to occur more quickly. After the trajectory leaves $\Omega_i$, one of the
saddle points $s_{ij}$ is identified with this pathway using, for
instance, the nudged elastic band 
method~\cite{Henkelman:2000,Henkelman:2000b}, and the low temperature
reaction rate is computed using \eqref{e:kramers} with
$\beta = \beta^{\lo}$. 
This is then repeated, with a new initial condition chosen in $\Omega_i$.  Throughout, the cumulative simulation
time is recorded. 

Other saddle point search algorithms have been proposed, including the
Dimer method and the string method~\cite{Olsen:2004du,E:2007ee}. In
our analysis, the key property that we need to hold true for all of our search methods is the following.  Let
\begin{equation}
  N_{ij}(t) = \text{Number of times saddle $s_{ij}$ has been found by
    time $t$}.
\end{equation}
Then for fixed $i$, during a saddle point search, the $N_{ij}(t)$ are
independent, with respect to $j$, Poisson processes. We prove
below that this holds for a carefully performed saddle point search
via integration of \eqref{e:od}.

This article is organized as follows. We describe the saddle point
search in detail in Section~\ref{s:algorithm} below, and prove some of
its properties, including the above condition on $N_{ij}(t)$, in
Section~\ref{s:properties} below.  In Section~\ref{s:estimators} we
introduce stopping criteria for the saddle point search, and in
Section~\ref{s:analysis} we analyze these criteria.
Section~\ref{s:estimates} contains proofs of some of the estimates in
Section~\ref{s:analysis}.  In Section~\ref{s:discussion} we make some
concluding remarks.

\section{Notation and saddle point search
  algorithm}\label{s:algorithm}

Here and throughout $(X_t)$ is Brownian dynamics, {that is, a
  stochastic process satisfying~\eqref{e:od}.}  For simplicity we fix
a single metastable set $\Omega \equiv \Omega_i$ and suppress the
index $i$ in all of our notations from the Introduction. For our
purposes, $V$ is smooth, and $\Omega$ is an (open) basin of attraction
of $V$ with respect to the gradient dynamics $dy/dt = -\nabla V(y)$.
We assume that $\partial \Omega$ is partitioned into finitely many
disjoint (measurable) subsets, called {\em pathways} and labeled
$1,2,\ldots, N$, such that each pathway $j$ contains a unique saddle
point $s_j$ of~$V$. {When $(X_t)$ leaves $\Omega$, it must exit
  through one of the pathways $1,2,\ldots,N$.}

The algorithm, as well as our analysis, depends heavily on the {quasistationary distribution} (QSD) for $(X_t)$ in $\Omega$, which
we denote by $\nu$.  The QSD $\nu$ is a probability measure that is
locally invariant for $(X_t)$, in the sense that it is invariant
conditionally on the event that $(X_t)$ remains in $\Omega$:
\begin{definition}
  The QSD for $(X_t)$ in $\Omega$ is a probability measure $\nu$
  supported in $\Omega$ such that for all $t > 0$,
  \begin{equation*}
    \nu(\cdot) = {\mathbb P}(X_t \in \cdot\,|\,
    X_0 \sim \nu,\,X_s \in \Omega \hbox{ for all } s \in [0,t]).
  \end{equation*}
\end{definition}
Of course $\nu$ depends on $\Omega$, but for simplicity we do not
indicate this explicitly. {It has been shown~\cite{LeBris:2012et} that
  $\nu$ exists, is unique, and satisfies
  \begin{equation}\label{e:QSD}
    \nu(A) = 
    \lim_{n\to \infty} {\mathbb P}(X_t \in A\,|\,X_s \in A \hbox{ for }s \in [0,t]), \qquad \hbox{ for all } A \subset \Omega.
  \end{equation}
  Moreover this convergence is exponentially fast, uniformly in $A$.
  Equation~\eqref{e:QSD} leads to simple algorithms for sampling
  $\nu$, based on the idea that a sample can be obtained from the
  endpoint of a trajectory of $(X_t)$ that has remained in $\Omega$
  for a sufficiently long time; see~\cite{Binder:2015gu} for details.}

We are now ready to state the high temperature saddle point search
algorithm. {Versions of this algorithm have been used previously; see for
  instance~\cite{Chill:2014es} and references therein. The search runs
  at a user-specified ``high'' (inverse) temperature
  $\beta^{\hi}$. Below we write $\nu$ for the QSD in $\Omega$ at
  temperature $\beta = \beta^{\hi}$.  We also write
  \begin{equation*}
    H(t) = \begin{cases} 0, &t < 0\\1, & t \ge 0\end{cases}
  \end{equation*}
  for the Heaviside unit step function.  }
{\begin{algorithm}\label{a:sps} Set $N_j(t) \equiv 0$ for $t \ge 0$
    and $j = 1,\ldots,N$. Let $M$ be the current cycle of the
    algorithm, and $\tsim$ the simulation clock.  Initialize $M=1$
    and $\tsim = 0$, and iterate the following: \vskip5pt
    \begin{itemize}
    \item[1.] Generate a sample $x_M$ from $\nu$. The simulation clock
      $\tsim$ is stopped during this step.
    \item[2.] Starting at $X_0 = x_M$, evolve $(X_t)$ at
      $\beta = \beta^{\hi}$ until it first leaves $\Omega$, say at
      time $t=\tau^{(M)}$ through pathway $I^{(M)}$. The simulation
      clock $\tsim$ is running during this step, and the stopping
      criterion is continuously checked.  If at some time $\tsim$ 
      the criterion is met, the algorithm stops.
\vskip5pt
    \item[3.] If $I^{(M)} = j$, update $N_j(t) = N_j(t) + H(t-\tsim)$
      for $t \ge 0$ and record the saddle point~$s_j$. Then update
      $M = M+1$. The simulation clock $\tsim$ is stopped during this step.
    \end{itemize}
  \end{algorithm}
  It is not necessary to know $N$, and the pathways can be
  given labels according to the order in which they are found. The
  simulation clock is cumulative, and it only increases in Step 2.  In
  particular, during the $M$-th cycle of the algorithm, $\tsim$
  increases by $\tau^{(M)}$. The stopping criterion will be described
  in Section~\ref{s:estimators}.  Below we write $\tsim$ for the final
  value of the simulation clock in the algorithm, that is, its value
  when the simulation is stopped. To refer to a generic simulation
  clock time we write $t$. Thus, $0 \le t \le \tsim$ and when the
  algorithm stops, $N_j(t)$ is the number of times an exit through
  pathway $j$ has been observed by time $t$. Below we write $N_j(t)$
  for its final value when the algorithm stops. We will also use the
  following notations:
  \begin{equation}\label{e:chi}
    \chi_j(t) = {\mathbbm{1}}_{N_j(t) \ge 1}, 
    \qquad N(t) = \sum_{j=1}^N N_j(t).
  \end{equation}
  That is, $\chi_j(t) = 1$ if an exit through pathway $j$ has been
  observed at least once by time $t$, and is $0$ otherwise; $N(t)$ is
  the total number of exits observed by time $t$.}

\section{Properties of the saddle point search}\label{s:properties}
Our first result follows immediately from properties of the QSD
{established in~\cite{LeBris:2012et}.}
\begin{theorem}\label{t:QSD}
  {Suppose that in Step 1 in the $M$-th cycle of
    Algorithm~\ref{a:sps}, $x_M$ is a random variable with
    distribution $\nu$. Then:}
  \begin{itemize}
  \item[](i)\,\, $\tau^{(M)}$ is exponentially distributed with mean
    $\kappa^{-1}$: ${\mathbb P}(\tau^{(j)} > t) = \exp(-\kappa t)$,
  \item[](ii)\,\, $\tau^{(M)}$ and $I^{(M)}$ are independent.
  \end{itemize}
\end{theorem} {Theorem~\ref{t:QSD} then leads to the following.}
\begin{theorem}\label{t:poisson}
  {Suppose that in Step 1 of Algorithm~\ref{a:sps}, $x_1,x_2,\ldots$
    are iid with common distribution $\nu$. Then:}
  \begin{itemize}
  \item[] (i) \,\, $\{N(t)\}_{0 \le t \le \tsim}$ is a Poisson process
    with parameter $\kappa$, \vskip5pt
  \item[] (ii) \,\, $\{N_j(t)\}_{0 \le t \le \tsim}^{j=1,\ldots,N}$,
    are independent Poisson processes with parameters
    \begin{equation}
      \kappa_j := \kappa\,p_j, \qquad p_j := {\mathbb P}(I^{(1)} = j).
    \end{equation} 
  \end{itemize}
\end{theorem}

\begin{proof}
  Let $({\tilde N}(s))_{s \ge 0}$ be a Poisson process with parameter
  $\kappa$, which we denote by ${\tilde N}(s)$ for brevity.  Label
  each arrival time of ${\tilde N}(s)$ with a pathway $j$ according to
  the distribution $p_j$, independently of the other arrival times,
  and let ${\tilde N}_j(s)$ be the process with arrivals labeled
  by~$j$. Then for $r,s \ge 0$ and $m_1,\ldots,m_N \ge 0$,
  \begin{align}\begin{split}\label{e:poisson}
      &{\mathbb P}\left(\bigcap_{j=1}^N \left\{{\tilde N}_j(r+s)-{\tilde N}_j(r) = m_j\right\}\right) \\
      &={\mathbb P}\left(N(r+s)-N(r) = \sum_{j=1}^N m_j\right){m_1+\ldots+m_N\choose m_1,\ldots,m_N} \prod_{j=1}^N p_j^{m_j}\\
      &= \prod_{j=1}^N \frac{e^{-\kappa p_js}(\kappa
        p_js)^{m_j}}{m_j!}.\end{split}
  \end{align}
  By summing over all $m_i \ge 0$ for $i \ne j$ in the last expression
  above, we see that for fixed $r,s \ge 0$, the increment
  ${\tilde N}_j(r+s)-{\tilde N}_j(r)$ is Poisson distributed with mean
  $\kappa p_j s$. ${\tilde N}_j(s)$ also inherits independent
  increments from ${\tilde N}(s)$.  This shows that ${\tilde N}_j(s)$
  is a Poisson process with parameter $\kappa_j = \kappa
  p_j$.
  Moreover,~\eqref{e:poisson} shows that ${\tilde N}_j(s)$,
  $j=1,\ldots,N$, are {independent}.

  Let us now relate $({\tilde N}(s))_{s \ge 0}$ with
  $(N(s))_{0 \le s \le \tsim}$. For fixed $s \in [0,\tsim]$, the time
  marginal $N(s)$ is the largest $m$ such that
  $\tau^{(1)}+\ldots+\tau^{(m)} \le s$.  Together with part~{\em (i)}
  of Theorem~\ref{t:QSD}, this shows that on $[0,\tsim]$,
  $(N(s))_{0 \le s \le \tsim}$ and $({\tilde N}(s))_{s \ge 0}$ are
  Poisson processes with the same law.  By part~{\em (ii)} of
  Theorem~\ref{t:QSD}, it follows that the multivariate processes
  $(N_j(s))_{0 \le s \le t_{sim}}^{j=1,\ldots,N}$ and
  $({\tilde N}_j(s))_{0 \le s \le \tsim}^{j=1,\ldots,N}$ have the same
  law. This establishes the result.
\end{proof}

\section{Chill type estimators and stopping
  criteria}\label{s:estimators}

The purpose of the high temperature saddle point search
(Algorithm~\ref{a:sps}) is to locate ``enough'' of the low-temperature
rate corresponding to the metastable set $\Omega$.  More precisely, at
a low temperature corresponding to $\beta = \beta^{\rm lo}$, the first
exit time of $X_t$ from $\Omega$ is approximately exponentially
distributed with mean $(k_1+\ldots+k_N)^{-1}$, where
$k_j = k_j^{\rm lo}$ is given by the Eyring-Kramers
law~\eqref{e:kramers} at $\beta = \beta^{\rm lo}$ (recall the
subscript $i$ has been suppressed). See~\cite{Berglund:2015} and
references therein for rigorous results in this direction. The $k_j$'s
are then exponential rates associated with leaving $\Omega$ through
pathway $j$ at low temperature $\beta^{\rm lo}$.  {The proportion of
  low temperature rate found by time $t$ in Algorithm~\ref{a:sps}~is}
\begin{equation}\label{e:R} {R}(t) := \frac{\sum_{j=1}^N \chi_j(t)
    k_j}{\sum_{j=1}^N k_j}.
\end{equation}
The expected value of $R(t)$ is
\begin{equation}\label{e:Rbar} {\mathbb E}[R(t)] = {\bar R}(t) :=
  \frac{\sum_{j=1}^N p_j(t) k_j}{\sum_{j=1}^N k_j},
\end{equation}
where
\begin{equation}
  \qquad p_j(t):={\mathbb E}[\chi_j(t)] = 1-\exp(-\kappa_j t).
\end{equation}
Here $\kappa_j$ is defined as in Theorem~\ref{t:poisson} at temperature 
$\beta = \beta^{\hi}$.  The idea
behind Chill-type estimators is that when ${R}(t)$ is sufficiently
close to $1$, the high temperature saddle point search can stop. There
are two obstacles to this idea.

The first is that, at any time during Algorithm~\ref{a:sps}, it is
unlikley that all saddle points have been found.  This problem is
remedied by replacing $k_j$ in~\eqref{e:R} with $\chi_j(t) k_j$, which
is computable once pathway $j$ has been found during the
simulation. The second obstacle is that an exact formula for
$p_j(t):= {\mathbb E}[\chi_j(t)]$ will not be known in practice.
Chill-type estimators overcome the latter obstacle by using one of the
following approximations:
\begin{align}\begin{split}\label{e:approxs}
    &{\tilde p}_j(t) :=
    1 - \exp[-k_j^{hi}t], \qquad k_j^{\hi} \hbox{ given by the Eyring-Kramers law at }\beta = \beta^{\hi},\\
    &{\hat p}_j(t) := 1 - \exp[-{\hat N}_j(t)], \qquad {\hat N}_j(t)
    := \begin{cases} N_j(t), & N_j(t) \ge 2, \\ 0, &
      \hbox{else}\end{cases}.\end{split}
\end{align} 
We have used the superscript $^{\hi}$ to indicate that the rate
in~\eqref{e:approxs} is computed at temperature $\beta^{\hi}$ (whereas
$k_j$ is computed at low temperature $\beta^{\lo}$). Also note that
${\tilde p}_j(t)$ is a physical estimate of ${\mathbb E}[\chi_j(t)]$
based on Eyring-Kramers, while ${\hat p}_j(t)$ {is a (biased) Monte
  Carlo estimator.}  From~\eqref{e:approxs} we obtain the following
estimators for ${R}(t)$:
\begin{equation}\label{e:Rt} {\tilde R}(t) := \frac{\sum_{j=1}^N
    {\tilde
      p}_j(t)\chi_j(t)k_j}{\sum_{j=1}^N \chi_j(t)k_j},\qquad {\hat
    R}(t) := \frac{\sum_{j=1}^N {\hat
      p}_j(t)\chi_j(t)k_j}{\sum_{j=1}^N \chi_j(t)k_j}.
\end{equation}
$R(t)$, ${\tilde R}(t)$, and ${\hat R}(t)$ are all random,
while ${\bar R}(t)$ is deterministic.  Both ${\tilde R}(t)$ and ${\hat R}(t)$
are explicitly computable at time 
$t$ during the saddle point search.  {See~\cite{Chill:2014es} for further
  discussion of ${\tilde R}(t)$.  To our knowledge ${\hat R}(t)$ has
  not appeared before in the literature.  We emphasize that
  ${\hat R}(t)$ may be used at any temperature $\beta^{\hi}$, while
  ${\tilde R}(t)$ is limited by 
  the fact that it gives reasonable 
  estimates of $R(t)$ only at
  (relatively low) 
  temperatures where the Eyring-Kramers law holds.} 


After choosing ${\tilde R}(t)$ (resp. ${\hat R}(t)$) as the preferred
estimator, the stopping criterion can now be defined as follows: for a
user-specified parameter $\epsilon > 0$, stop Algorithm~\ref{a:sps} in
Step 3 if and only if
\begin{equation} {\tilde R}(t) > 1-\epsilon \qquad \hbox{(resp. }
  {\hat R}(t) > 1-\epsilon \hbox{)}.
\end{equation}
In Section~\ref{s:analysis} we give rigorous estimates of the bias and
variance of the estimators ${\tilde R}(t)$ and ${\hat R}(t)$. Such
estimates will show that, {as $t$ increases}, when the algorithm
stops, on average at least $(1-\epsilon)\%$ of the low temperature
rate has been found.

\section{Analysis}\label{s:analysis}

The approximation ${\tilde p}_j(t)$ of $p_j(t)$ is usually considered
valid when $\beta^{\rm hi} \ll V(s_j)-V(m)$, with $m$ the minimizer of
$V$ in $\Omega$.  To the authors' knowledge, rigorous results are
scarce except when $s_j = \argmin_{s_1,\ldots,s_N}
V(s_j)-V(m)$; see~\cite{Berglund:2015} 
and references therein.
However, the following is a consequence of results
in~\cite{Aristoff:2014ch}:
\begin{theorem}\label{t:ratios}
  {Suppose $\Omega = (a,b)$ is an interval and $V$ is a Morse
    potential. Then for each $t > 0$,
    \begin{equation}
      \frac{1-{\tilde p}_j(t)}{1-p_j(t)} = 1 + O(1/\beta^{\hi}) \hbox{ as }\beta^{\hi}\to \infty, 
      \qquad j = 1,2.
    \end{equation}}
\end{theorem}
\begin{proof}
  An examination of the proof of Theorem 4.1 of~\cite{Aristoff:2014ch}
  shows that for $j=1,2$,{
    \begin{equation*}
      k_j^{\rm hi}/\kappa_j = 1 + O(1/\beta^{\hi}) 
      \qquad \hbox{ as } \beta^{\hi} \to \infty,
    \end{equation*}
    where $k_j^{\hi}$ is as in~\eqref{e:approxs}, and
    $\kappa_j$ is as in 
    Theorem~\ref{t:poisson} at temperature $\beta = \beta^{\hi}$.}
  The result follows.
\end{proof}
We next examine the approximation ${\hat p}(t)$ of $p(t)$.
\begin{theorem}\label{t:rate_est}
  Conditionally on $N(t) \ge 1$, ${\hat N}_j(t)$ is an unbiased
  estimator for $\kappa_j t$:
  \begin{equation} {\mathbb E}[{\hat N}_j(t)\,|\, N(t) \ge 1] =
    \kappa_j t.
  \end{equation} 
  Also conditionally on $N(t) \ge 1$, ${\hat p}_j(t)$ is a
  conservative estimate of $p_j(t)$:
  \begin{equation} {\mathbb E}[{\hat p}_j(t)\,|\, {N}_j(t) \ge 1] \le
    p_j(t).
  \end{equation}
\end{theorem}
\begin{proof}
  Recall that $N_j(t)$ is a Poisson process with parameter $\kappa_j$.
  {Thus,
    \begin{align*}
      {\mathbb E}[{\hat N}_j(t)\,|\, N_j(t) \ge 1] &= 
                                                     \left(1-e^{-\kappa_j t}\right)^{-1}\sum_{n=2}^\infty n \frac{(\kappa_j t)^n e^{-\kappa_j t}}{n!} \\
                                                   &= \frac{\kappa_j t}{1-e^{-\kappa_j t}}\sum_{n=1}^\infty \frac{(\kappa_j t)^{n} e^{-\kappa_j t}}{n!} = \kappa_j t.
    \end{align*}}
  Since $x \mapsto 1-e^{-x}$ is a concave function, the second
  statement of the theorem follows from Jensen's inequality.
\end{proof}
The reason that we consider conditional expectations in
Theorem~\ref{t:rate_est} is that Algorithm~\ref{a:sps} cannot stop
before $N(t) \ge 1$. Thus, we want estimates conditioned on
that event. We call ${\hat p}_j(t)$ a conservative estimate for
$p_j(t)$ because it is a {lower bound on average, so that using
  ${\hat p}_j(t)$ in place of $p_j(t)$ leads to a larger average
  stopping time for Algorithm~\ref{a:sps}.}

Before proceeding we define, for a real-valued random variables $X$
and~$Y$,
\begin{equation}\label{e:bias_mse}
  \Bias(X,Y) := {\mathbb E}[X-Y], 
  \qquad \MSE(X,Y) := \Bias(X,Y)^2 
  + \Var(X).
\end{equation}
Observe that the mean square error is not symmetric in its arguments.
\begin{theorem}\label{t:error_bounds}
  Write $q_j(t) = 1-p_j(t) = \exp[-\kappa_j t]$ and
  $K = k_1 + \ldots + k_N$.  For the estimator ${\tilde R}(t)$,
  \begin{align}\begin{split}
      &\abs{\Bias(\tilde{R}(t),{R}(t))} \le N \max_j
      \abs{\Bias({\tilde p}_j(t) ,
        p_j(t))} + \frac{K}{\min_j k_j}{\bar R}(t) \max_j q_j(t),\\
      &\Var(\tilde {R}(t)) \leq  {4\frac{K^2}{\min_j k_j^2}}{\bar R}(t)^2 \max_{j} q_j(t),\\
      &\MSE(\tilde {R}(t), { R}(t)) \leq 2N^2 \max_j \MSE(\tilde
      p_j(t),
      p_j(t)) \\
      &\qquad\qquad\qquad\qquad\qquad+ {\frac{K^2}{\min_{j}
          k_j^2}}\paren{2\max_{j} q_j(t) + 4 }{\bar R}(t)^2 \max_{j}
      q_j(t).
    \end{split}
  \end{align}
  For the estimator ${\hat R}(t)$,
  \begin{align}\begin{split}
      &\abs{\Bias(\hat{R}(t), {R}(t))}\leq N \max_j
      \abs{\Bias(\hat{p}_j(t), p_j(t))} +
      \frac{K}{\min_j k_j}{\bar R}(t)\max_j q_j(t) ,\\
      &\Var(\hat{R}(t))\leq \frac{2K^2}{\min_j k_j^2}{\bar R}(t)^2
      \max_j q_j(t)
      + \paren{1 + 2N^2 \max_j q_j(t)}\max_j\Var(\hat{p}_j(t)),\\
      &\MSE(\hat{R}(t),{R}(t)) \le \paren{1 + N^2 + 2N^2 \max_j q_j(t)}\max_j\MSE(\hat{p}_j(t),p_j(t))\\
      &\qquad\qquad\qquad\qquad\qquad+ \frac{4K^2}{\min_j k_j^2}{\bar
        R}(t)^2\paren{1 + \max_j q_j(t)}\max_j q_j(t).
    \end{split}
  \end{align}
  Here, all maxima and minima are taken over $j \in \{1,\ldots,N\}$.
\end{theorem}
\begin{proof} We give proofs in Section~\ref{s:estimates} below.
\end{proof}


We note that some of the bounds in Theorem~\ref{t:error_bounds} have
been loosened so that simpler expressions are obtained. This will
become clear in the derivation of the bounds in
Section~\ref{s:estimates} below. We highlight that the bias is bounded
by the bias of the estimate of $p_j(t)$, together with another term
representing an ``inherent'' bias associated with ${\bar R}(t)$. This
second term may be approximated by noting that $|{\bar R}(t)|< 1$ for
all $t$ and, due to Theorem~\ref{t:ratios}, we expect $q_j(t)$ can be
estimated by the known function ${\tilde p}_j(t)$ or $\hat{p}_j(t)$.

\section{Estimates}\label{s:estimates}

In this section we give a proof of Theorem~\ref{t:error_bounds}.
Recall that $q_j(t) := 1- p_j(t)$ and $K := \sum_{j=1}^N k_j$ is the
total reaction rate.  For brevity, we will sometimes suppress the $t$
dependence in our expressions. Also, all sums are over $1,\ldots,N$
unless otherwise indicated.

\subsection{Preliminary Calculations}
Observe that
\begin{equation*}
  \Bias({\tilde R}(t),R(t)) = 
  \Bias({\tilde R}(t),{\bar R}(t)), 
  \qquad \MSE({\tilde R}(t),R(t)) = 
  \MSE({\tilde R}(t),{\bar R}(t))
\end{equation*}
and similarly for ${\hat R}(t)$; this fact will be used below without
comment.  There are a few expressions that will show up repeatedly in
the analyses of both ${\tilde R}$ and ${\hat R}$.  We analyze them
here for simplicity.  Let
\begin{equation}
  \label{e:chit}
  \xi_i = k_i + \sum_{m\neq i} k_m\chi_m 
\end{equation}
We make the following calculations:
\begin{subequations}
  \begin{gather}
    k_i \leq \xi_i \leq K\\
    \E[\xi_i] = k_i + \sum_{m\neq i} p_m k_m = K - \sum_{m\neq i}
    q_mk_m
  \end{gather}
\end{subequations}
A lower bound on this can be obtained from Jensen's inequality,
\begin{equation}
  \label{e:Exi1_lower}
  \E[\xi_i^{-1}]\geq \E[\xi_i]^{-1} = \frac{1}{K - \sum_{m\neq i}
    q_mk_m}\geq \frac{1}{K} +\frac{1}{K^2} \sum_{m \neq i} k_m
  q_m
\end{equation}
while an upper bound can be obtained from the Edmunson-Madansky
inequality,
\begin{equation}
  \label{e:Exi1_upper}
  \E[\xi_i^{-1}]\leq \frac{1}{k_i}\frac{K -
    \E[\xi_i]}{K- k_i} + \frac{1}{K}\frac{
    \E[\xi_i]-k_i}{K- k_i}=\frac{1}{K} + \frac{1}{k_i K}\sum_{m \neq
    i} k_m q_m
\end{equation}
In the same way,
\begin{equation}
  \label{e:Exi2_lower}
  \E[\xi_i^{-2}]\geq \E[\xi_i]^{-2} = \frac{1}{\paren{K - \sum_{m\neq i}
      q_mk_m}^2}\geq \frac{1}{K^2} + \frac{2}{K^3}\sum_{m \neq i}
  q_mk_m
\end{equation}
and
\begin{equation}
  \label{e:Exi2_upper}
  \E[\xi_i^{-2}]\leq \frac{1}{k_i^2}\frac{K - \E[\xi_i]}{K-
    k_i} + \frac{1}{K^2}\frac{
    \E[\xi_i]-k_i}{K- k_i}= \frac{1}{K^2} + \frac{K+k_i}{k_i^2 K^2} \sum_{m\neq i}
  q_m k_m
\end{equation}
Therefore,
\begin{equation}
  \label{e:varxi}
  \Var(\xi_i^{-1}) \leq \paren{ \frac{K+k_i}{k_i^2 K^2} -
    \frac{2}{K^3}}\sum_{m\neq i} q_m k_m \leq \frac{2}{K k_i^2}\sum_{m\neq i} q_m(t) k_m,
\end{equation}
where we have lost some of the estimate in the last inequality for the
sake of conciseness.

\subsection{Estimates for ${\tilde R}$}
Below it is useful to notice that
\begin{equation} {\tilde R}(t) = \sum_{i=1}^N\frac{{\tilde
      p}_i(t)\chi_i(t) k_i}{k_i + \sum_{m \ne i} \chi_m(t)k_m} =
  \sum_i \frac{{\tilde p}_i \chi_i k_i}{\xi_i}.
\end{equation}

\subsubsection{Bias}
We begin with the direct calculation
\begin{align*}
  \E[{\tilde R} - {\bar R}] &= \sum_{i=1}^N \E\bracket{\frac{\chi_i \tilde p_i
                              k_i}{\xi_i} - \frac{\chi_i k_i}{K}} \\
                            &= \sum_{i=1}^{N} (\tilde p_i -
                              p_i) \E\bracket{\frac{\chi_i k_i}{\xi_i}}+ \sum_{i=1}^N \E\bracket{\frac{\chi_i
                              p_i k_i}{\xi_i} - \frac{\chi_i k_i}{K}}\\
                            &= \sum_{i=1}^{N} (\tilde p_i -
                              p_i) \E\bracket{\frac{\chi_i k_i}{\xi_i}} 
                              + \sum_{i=1}^N \underbrace{\E\bracket{\frac{K p_i}{\xi_i} -
                              1}}_{\equiv b_i}\frac{p_i k_i}{K}.
\end{align*}
Using \eqref{e:Exi1_lower} and \eqref{e:Exi1_upper},
\begin{equation*}
  \frac{1}{ K}\sum_{m \neq
    i} k_m q_m - q_i\leq b_i \leq \frac{1}{k_i}\sum_{m \neq
    i} k_m q_m - q_i.
\end{equation*}
Thus,
\begin{equation*}
  \left|\sum_{i=1}^N b_i\frac{p_i k_i}{K}\right| 
  \le \sum_{i=1}^N \paren{\sum_{j=1}^N
    \frac{k_j}{k_i} q_j  }\frac{p_i(t) k_i}{K}\\\
  \leq
  \frac{K \max_{j} q_j(t) }{\min_{j} k_j}\, {\bar R}(t).
\end{equation*}
Combining the above expressions gives
\begin{equation}
  \label{e:bias2_est}
  \abs{\Bias({\tilde R}(t),{\bar R}(t))}
  \leq N \max_i \abs{\tilde p_i(t) -
    p_i(t)} + \frac{K\max_i q_i(t)}{\min_i k_i}{\bar R}(t) .
\end{equation}

\subsubsection{Variance}
For the variance, we first write
\begin{equation}
  {\tilde R} -\E[{\tilde R}] =
  \sum_{i=1}^N \paren{\frac{\chi_i}{\xi_i} -
    \E\bracket{\frac{\chi_i}{\xi_i}}}{\tilde p}_i k_i. 
\end{equation}
Hence,
\begin{equation}
  \Var({\tilde R}(t)) = \sum_{i,j=1}^N k_i k_j {\tilde p}_i {\tilde p}_j
  \underbrace{\Cov\paren{\frac{\chi_i}{\xi_i} ,
      \frac{\chi_j}{\xi_j} }}_{\equiv v_{ij}}.
\end{equation}
Since $v_{ij}\leq \sqrt{v_{ii}}\sqrt{v_{jj}}$, it will be sufficient
for us to analyze the diagonal terms.  By Theorem~\ref{t:poisson},
$\chi_i$ and $\xi_i$ are independent. Thus
\begin{equation}
  v_{ii} = \E[\chi_i]^2 \Var(\xi_i^{-1}) + \E[\xi_i^{-1}]^2 \Var(\chi_i)
  +  \Var(\xi_i^{-1}) \Var(\chi_i).
\end{equation}
Using \eqref{e:Exi2_upper} and \eqref{e:varxi},
\begin{equation}
  \begin{split}
    v_{ii} &\leq p_i\Var(\xi_i^{-1})  + p_iq_i \E[\xi_i^{-2}]\\
    &\leq p_i \paren{ \frac{K+k_i}{k_i^2 K^2} -
      \frac{2}{K^3}}\sum_{m\neq i} q_m k_m + p_i
    q_i\paren{\frac{1}{K^2} + \frac{K+k_i}{k_i^2 K^2}
      \sum_{m\neq i} q_m k_m}\\
    &\leq \frac{p_i q_i}{K^2} +\frac{4 p_i}{k_i^2 K}\sum_{m\neq i}
    q_mk_m\leq \frac{4 p_i}{k_i^2}\max_{j} q_j\leq \frac{4
    }{k_i^2}\max_{j}
    q_j(t)\\
    &\leq \frac{4}{\min_j k_j^2}\max_{j} q_j(t).
  \end{split}
\end{equation}
We have made some sacrifices in the last inequalities in order to
obtain a more concise expression.  Consequently,
\begin{equation}\label{e:var2_est}
  \begin{split}
    \Var({\tilde R}(t)) &\leq \sum_{i,j=1}^N k_i k_j {\tilde p}_i(t)
    {\tilde p}_j(t)\sqrt{v_{ii}}\sqrt{v_{jj}}\\
    & \leq {\frac{4K^2}{\min_i k_i^2}}{\bar R}(t)^2 \max_{i} q_i(t).
  \end{split}
\end{equation}

\subsubsection{MSE}
Combining \eqref{e:bias2_est} and \eqref{e:var2_est}, we then obtain
\begin{equation}
  \label{e:mse2_est}
  \begin{split}
    \MSE({\tilde R}(t), {\bar R}(t)) &\leq 2N^2 \max_i \abs{\tilde
      p_i(t) -
      p_i(t)}^2\\
    & + {\frac{K^2}{\min_{i} k_i^2}}\paren{2\max_{i} q_i(t) + 4 }{\bar
      R}(t)^2 \max_{i} q_i(t).
  \end{split}
\end{equation}
In this calculation, we see that the mean square error may ultimately
be dominated by how well the $\tilde p_i$ approximate the $p_i$.

\subsection{Estimates for ${\hat R}$}

We begin by noting that, since $\hat{p}_j(t)=0$ if $\chi_j(t) \ne 1$,
\begin{equation}
  \label{e:chill3_v2}
  {\hat R}(t) = \sum_{j} \frac{\hat p_j(t) k_j}{k_j + \sum_{m\neq j} \chi_m(t) k_j}.
\end{equation}

\subsubsection{Bias}
We begin by writing
\begin{equation} {\hat R} - {\bar R} = \sum_{i=1}^N (\hat{p}_i - p_i)
  \frac{k_i}{\xi_i} + \sum_{i=1}^N\frac{k_i p_i}{\xi_i} - \frac{k_i
    p_i}{K}
\end{equation}
so that, after taking an expectation,
\begin{equation}
  \E[{\hat R} - {\bar R}] = \sum_{i=1} \E\bracket{ (\hat{p}_i - p_i)
    \frac{k_i}{\xi_i} } + \sum_{i=1}^N \paren{\E\bracket{\frac{K}{\xi_i}}-1}\frac{k_i p_i}{K}.
\end{equation}
Hence,
\begin{equation}
  \label{e:bias3}
  \abs{\Bias({\hat R}(t), {\bar R}(t))}\leq N \max_i \abs{\Bias(\hat{p}_i(t), p(t))} +
  \frac{K}{\min_i k_i}{\bar R}(t)\max_i q_i(t),
\end{equation}
and we see that the observed bias is controlled by the biases of the
approximate probabilities, $\hat{p}_i$, and the inherent bias of the
Chill type estimators.

\subsubsection{Variance}
For the variance, we have
\begin{equation}
  \Var({\hat R}) = \sum_{i,j=1}^N k_i k_j
  \underbrace{\Cov \paren{\frac{\hat{p}_i}{\xi_i},
      \frac{\hat{p}_j}{\xi_j}}}_{\equiv \hat{v}_{ij}}.
\end{equation}
As before, we only need to study the diagonal entries, and use
Theorem~\ref{t:poisson} to obtain
\begin{equation}
  \begin{split}
    \hat{v}_{ii} &= \E[\hat{p}_i]^2\Var(\xi_i^{-1}) + \E[\xi_i^{-1}]^2
    \Var(\hat p_i) + \Var(\hat p_i) \Var(\xi_i^{-1})\\
    &\leq \Var(\xi_i^{-1})  + \E[\xi_i^{-2}] \Var(\hat{p}_i)\\
    &\leq \frac{2}{\min_i k_i^2}\max_i q_i + \paren{\frac{1}{K^2} +
      \frac{2}{\min_i k_i^2} \max_i q_i}\Var(\hat{p}_i)\\
    &\leq \frac{2}{\min_i k_i^2}\max_i q_i + \paren{\frac{1}{K^2} +
      \frac{2}{\min_i k_i^2} \max_i q_i}\max_i\Var(\hat{p}_i).
  \end{split}
\end{equation}
We note that these estimates require full independence of $N_j(t)$ for
$j=1,\ldots,N$, not just independence of the $\chi_j(t)$.  Now,
\begin{equation}
  \label{e:var3}
  \Var({\hat R}(t))\leq \frac{2K^2}{\min_i k_i^2}{\bar R}(t)^2 \max_i q_i(t)
  + \paren{1 + 2N^2 \max_i q_i(t)}\max_i\Var(\hat{p}_i(t)).
\end{equation}

\subsubsection{MSE}
We can therefore express the mean square error of estimator ${\hat R}$
as
\begin{equation}
  \label{e:mse3_est}
  \begin{split}
    \MSE({\hat R}(t),{\bar R}(t))&\leq \frac{4K^2}{\min_i
      k_i^2}{\bar R}(t)^2\paren{1 + \max_i q_i(t)}\max_i q_i(t) \\
    &\quad+ \paren{1 + N^2 + 2N^2 \max_i
      q_i(t)}\max_i\MSE(\hat{p}_i(t),p_i(t)).
  \end{split}
\end{equation}

\begin{figure}
  \label{f:estimator}
\pgfplotsset{every axis/.append style={thick,line join=round}}
\begin{tikzpicture}
\begin{groupplot}[
	group style = {group size=3 by 1, horizontal sep=5pt},
	width=5cm,
	height=5cm,
	ymode=log,
	ymin=0.000006,
	grid=major,
	cycle list name=linestyles,
	unbounded coords=jump,
]
	\nextgroupplot[title={$n=-\tfrac{1}{2}$}, ylabel={1-Estimator}, xmax=600]
	\addplot table[x=time,y=exact]   {estimator_figure/quasiharmonic1.dat};
	\addplot table[x=time,y=chill_1] {estimator_figure/quasiharmonic1.dat};
	\addplot table[x=time,y=chill_2] {estimator_figure/quasiharmonic1.dat};

	\nextgroupplot[yticklabels={,,},title={$n=0$},xlabel={$t$}, xmax=500]
	\addplot table[x=time,y=exact]   {estimator_figure/harmonic.dat};
	\addplot table[x=time,y=chill_1] {estimator_figure/harmonic.dat};
	\addplot table[x=time,y=chill_2] {estimator_figure/harmonic.dat};
	
	\nextgroupplot[yticklabels={,,},title={$n=\tfrac{1}{2}$}, xmax=380]
	\addplot table[x=time,y=exact]   {estimator_figure/quasiharmonic2.dat};
	\addplot table[x=time,y=chill_1] {estimator_figure/quasiharmonic2.dat};
	\addplot table[x=time,y=chill_2] {estimator_figure/quasiharmonic2.dat};
	\legend{$1-R(t)$,$1-\tilde{R}(t)$,$1-\hat{R}(t)$}

\end{groupplot}
\end{tikzpicture}

  \caption{Comparison of the Chill type estimators $\tilde{R}(t)$ and
    $\hat{R}(t)$ to the true {expected proportion of the low
      temperature rate found, ${R}(t)$, on a test system that can
      deviate from the Erying-Kramers law.}}
\end{figure}
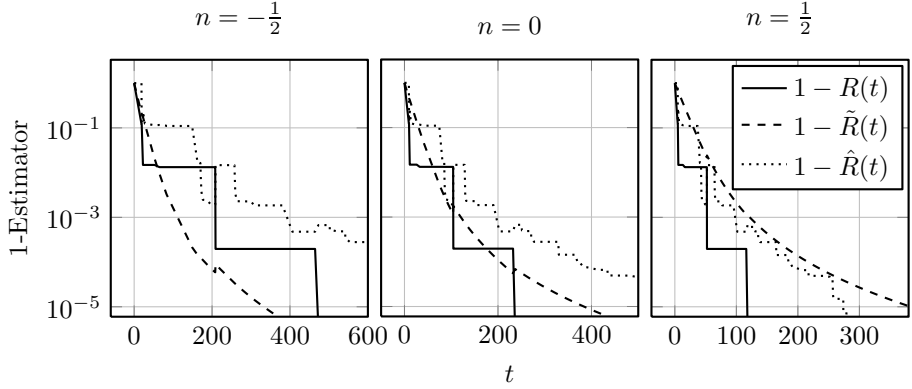

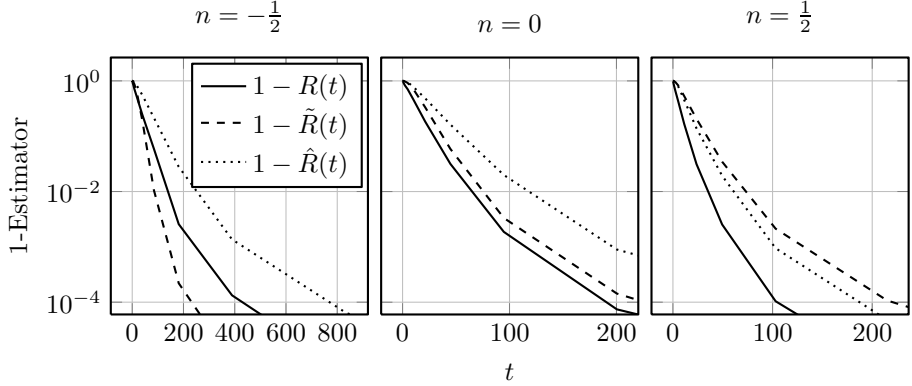
\begin{figure}
  \label{f:avg_estimator}
\pgfplotsset{every axis/.append style={thick,line join=round}}
\pgfplotscreateplotcyclelist{custom}{%
every mark/.append style={fill=gray}\\%
dashed,every mark/.append style={solid,fill=gray}\\%
dotted,every mark/.append style={solid,fill=gray}\\%
every mark/.append style={fill=gray},mark=square*\\%
every mark/.append style={fill=gray},mark=otimes*\\%
mark=star\\%
every mark/.append style={fill=gray},mark=diamond*\\%
densely dashed,every mark/.append style={solid,fill=gray},mark=*\\%
densely dashed,every mark/.append style={solid,fill=gray},mark=square*\\%
densely dashed,every mark/.append style={solid,fill=gray},mark=otimes*\\%
densely dashed,every mark/.append style={solid},mark=star\\%
densely dashed,every mark/.append style={solid,fill=gray},mark=diamond*\\%
}
\begin{tikzpicture}
\begin{groupplot}[
	group style = {group size=3 by 1, horizontal sep=5pt},
	width=5cm,
	height=5cm,
	ymode=log,
	ymin=0.00006,
	grid=major,
	cycle list name=custom,
	unbounded coords=jump,
]
	\nextgroupplot[title={$n=-\tfrac{1}{2}$}, ylabel={1-Estimator}]
  \addplot+ table[x=time,y=exact]   {avg_estimator_figure/quasiharmonic1.dat};
  \addplot+ table[x=time,y=chill_1] {avg_estimator_figure/quasiharmonic1.dat};
  \addplot+ table[x=time,y=chill_2] {avg_estimator_figure/quasiharmonic1.dat};
	\legend{$1-R(t)$,$1-\tilde{R}(t)$,$1-\hat{R}(t)$}

	\nextgroupplot[yticklabels={,,},title={$n=0$},xlabel={$t$}]
  \addplot+ table[x=time,y=exact]   {avg_estimator_figure/harmonic.dat};
  \addplot+ table[x=time,y=chill_1] {avg_estimator_figure/harmonic.dat};
  \addplot+ table[x=time,y=chill_2] {avg_estimator_figure/harmonic.dat};
	
	\nextgroupplot[yticklabels={,,},title={$n=\tfrac{1}{2}$}]
  \addplot+ table[x=time,y=exact]   {avg_estimator_figure/quasiharmonic2.dat};
  \addplot+ table[x=time,y=chill_1] {avg_estimator_figure/quasiharmonic2.dat};
  \addplot+ table[x=time,y=chill_2] {avg_estimator_figure/quasiharmonic2.dat};

\end{groupplot}
\end{tikzpicture}

  \caption{Comparison of the expected value of the Chill type estimators $\tilde{R}(t)$ and
    $\hat{R}(t)$ to the true {expected proportion of the low
      temperature rate found, ${R}(t)$, on a test system that can
      deviate from the Erying-Kramers law.}}
\end{figure}

\section{Discussion}\label{s:discussion}
We have considered three Chill type estimators and shown them to be
consistent.  Their biases are small, relative to their variances, and
thus we have good estimators of ${R}(t)$, the true fraction of the
observed rate in the system. They represent a significant
improvement over the original AKMC stopping criterion presented in
\cite{Xu:2008jk}.  Indeed, these prior
approaches attempted to estimate the fraction of the saddles observed
when, in fact, it is the fraction of the observed rate that is of
fundamental importance.

{As an example, we will compare the accuracy of both estimators using
  a test system that consists of saddle points $s_j$ corresponding to
  potential energy barriers $V(s_j)-V(m)= 1 + \frac{4}{19}j$, for
  $j=0,\ldots,19$. The test system has rates that obey a modified
  Arrhenius equation with the form:
  \begin{equation}
    {\tilde k}_j^{\rm hi} = \left(\frac{\beta^{\rm lo}}{\beta^{\rm hi}}\right)^{n}  g_j \exp[\beta V(s_j)-V(m)].
  \end{equation}
  Compare to equation~\eqref{e:kramers} (recall the subscript $i$ has
  been suppressed).  The variable $n$ controls how the rates deviate
  from an unmodified Arrhenius rate law. When $n=0$ the modified rates
  ${\tilde k}_j^{\rm hi}$ are equal to the unmodified rates
  $k_j^{\rm hi}$, while when $\beta^{\rm hi} < \beta^{\rm lo}$, the
  modified rates are larger (resp. smaller) than the unmodified rates
  if $n>0$ (resp. $n < 0$).}

{We use Algorithm~\ref{a:sps} on the test system with modified rates
  ${\tilde k}_j^{\rm hi}$. This means
  $(N_j(t))_{0 \le t \le \tsim}^{j=1,\ldots,N}$ are independent
  Poisson processes with parameters ${\tilde k}_j^{\rm hi}$.  To
  compute $R(t)$, we use~\eqref{e:R} and sample $\chi_j(t)$
  via~\eqref{e:chi}.  To compute ${\tilde R}(t)$ we use the unmodified
  Arrenius rates $k_j^{\rm hi}$ in equation~\eqref{e:approxs}.  For
  each of $R(t)$, ${\tilde R}(t)$ and ${\hat R}(t)$, the low
  temperature rates $k_j = k_j^{\rm lo}$ used in equations~\eqref{e:R}
  and~\eqref{e:Rt} are the same.  We take $g_j = 1$ for all $j$ and
  $\beta^{\rm hi}=2.5$, $\beta^{\rm lo}=10.0$. The variable $n$ was
  varied to compare the cases where the Erying-Kramers rates
  ${k}_j^{\rm hi}$ underestimate ($n=\frac{1}{2}$), overestimate
  ($n=-\frac{1}{2})$, and provide an exact estimate ($n=0$) of the
  modified rates ${\tilde k}_j^{\rm hi}$.  Results are shown in
  Figures~1 and~2.  The test system shows that $\tilde{R}(t)$
  can overestimate ${R}(t)$ if the Eyring-Kramers rate deviates from
  the true rate at $\beta^{\rm hi}$, while $\hat{R}(t)$ tends to
  provide a conservative estimate of $R(t)$.}

\bibliographystyle{plain}

\end{document}